\documentclass[a4paper]{article}
\RequirePackage{fullpage}

\usepackage{amsthm,amsmath}
\usepackage{subcaption}
\usepackage[usenames,dvipsnames]{xcolor}
\usepackage[pdftex,breaklinks,colorlinks,
    citecolor={BlueViolet}, linkcolor={Blue},urlcolor=Maroon]{hyperref}
\usepackage{tikz}
\usetikzlibrary{shapes}
\usetikzlibrary{arrows.meta}
\usetikzlibrary{quotes}
\usetikzlibrary{calc}
\usepackage{charter,eulervm}
\usepackage{enumerate, ifthen}
\usepackage[final,expansion=alltext,protrusion=true]{microtype}

\theoremstyle{plain} 
\newtheorem{theorem}{Theorem}[section]
\newtheorem{lemma}[theorem]{Lemma}

\newtheorem{proposition}[theorem]{Proposition}
\newtheorem{conjecture}[theorem]{Conjecture}
\newtheorem{reduction}{Rule}

\title{Minimum Sum Vertex Cover: \\Kernelization and Parameterized Algorithms}
\author{
  Yixin Cao\thanks{School of Computer Science and Engineering, Central South University, Changsha, China.  \texttt{jingyiliu@csu.edu.cn, jxwang@mail.csu.edu.cn}.} \thanks{Department of Computing, Hong Kong Polytechnic University, Hong Kong, China.  \texttt{yixin.cao@polyu.edu.hk}.
  }
  \and Ling Gai\thanks{Business School, University of Shanghai for Science and Technology, Shanghai, China.  lgai@usst.edu.cn.}
  \and Jingyi Liu\footnotemark[1]
  \and Jianxin Wang\footnotemark[1]
}
\date{}

\begin{document}
\maketitle
\begin{abstract}
  Given an ordering of the vertices of a graph, the cost of covering an edge is the smaller number of its two ends. The minimum sum vertex cover problem asks for an ordering that minimizes the total cost of covering all edges. We consider parameterized complexity of this problem, using the largest cost~$k$ of covering a single edge as the parameter. Note that the first $k$ vertices form a (not necessarily minimal) vertex cover of the graph, and the ordering of vertices after $k$ is irrelevant. We present a $(k^2 + 2 k)$-vertex kernel and an $O(m + 2^kk! k^4)$-time algorithm for the minimum sum vertex cover problem, where $m$ is the size of the input graph. Since our parameter~$k$ is polynomially bounded by the vertex cover number of the input graph, our results also apply to that parameterization.
\end{abstract}

\section{Introduction}

A \emph{vertex cover} of a graph~$G$ is a set of vertices such that every edge has at least one end in the set.
In other words, this vertex set ``covers'' all the edges of the graph.
There are many variations of this fundamental concept.
In the temporal setting, e.g., we take the vertices in a vertex cover one by one.
For the $i$th vertex~$v$, we assign the number~$i$ to all uncovered edges incident to~$v$, which are then covered \emph{by the vertex~$v$ with cost~$i$}.
The objective is to minimize the sum of costs assigned to all edges.
Given a graph~$G$, the \emph{minimum sum vertex cover} problem asks for an ordering of the vertices of~$G$ to minimize the total cost of all edges of~$G$, where the cost of an edge is the smaller number of its two ends in the ordering.

This problem was formally introduced by Feige et al.~\cite{FLT2004}.
They considered a more general problem, called the \emph{minimum sum set cover}, of which our problem is a special case.  They proved that the minimum sum vertex cover problem is APX-hard and admits a 2-approximation, which was improved to 16/9 recently by Bansal et al.~\cite{BBFT2023}.
Stankovi\'c~\cite{S2022} showed that the minimum sum vertex cover problem cannot be approximated within 1.014, assuming the Unique Games Conjecture.

In comparison, there was not much study on the parameterized complexity of the minimum sum vertex cover problem, and one explanation is about the parameter.
Note that to make a decision problem, the input has an additional integer~$w$, and the question becomes whether the input graph admits a solution with total cost at most~$w$.
The standard parameterization is uninteresting because in a yes-instance, $w$ cannot be smaller than the size of the graph.
One option, taken by Aute and Panolan~\cite{AP2024}, is to use the vertex cover number $\tau$, i.e., the size of a minimum vertex cover of the input graph.
They presented an algorithm running in time $\tau! (\tau + 1)^{2^\tau} n^{O(1)}$, where $n$ is the order of the input graph.
It starts with finding a minimum vertex cover~$S$ and guessing the ordering of the vertices in~$S$.
The vertices in $V(G)\setminus S$ is an independent set, and they can be partitioned into at most $2^\tau$ sets with respect to their neighborhoods in~$S$.  The key observation is that vertices in the same set (which are twins) appear consecutively in an optimal ordering.
The astronomical factor in the running time is for guessing the relative order of these sets.
Aute and Panolan~\cite{AP2024} also considered parameterization by the size of a minimum vertex cover of the complement of the input graph.

We take a natural alternative.  Let~$\sigma$ be an optimal ordering of~$V(G)$.  Taking the ends of all edges that are assigned the smaller number by~$\sigma$, we end with a vertex cover~$S$ of~$G$.
This vertex cover may or may not be minimal; e.g., consider the 19-vertex graph obtained from six disjoint claws by identifying one leaf from each claw.
It is an easy exercise to show that the vertices in~$S$ are at the beginning of~$\sigma$.
Indeed, what matters in~$\sigma$ is the ordering of vertices in~$S$, and other vertices can be appended to that in an arbitrary order.
We are thus motivated to use the size~$k$ of~$S$ as the parameter.

The parameters~$k$ and~$\tau$ are related.
On the one hand,
$\tau$ is a trivial lower bound for $k$.  On the other hand, a straightforward counting argument implies $k \le 2 \tau^2$, and we manage to show a nontrivial bound.
We conjecture that every graph has a solution in which $k= \tau+ O(\log \tau)$.

\begin{lemma}
  \label{lem:mvc and msvc}
  In any optimal solution of the minimum sum vertex cover problem, the maximum cost of an edge is $O(\tau^{1.5})$.
\end{lemma}

Our first result is a polynomial kernel.
Given an instance $(G, w, k)$, a {\em kernelization algorithm} produces in time polynomial in $n + k$ an equivalent instance $(G', w', k')$, called a \emph{kernel},
such that $k' \leq k$ and the order of $G'$ is bounded by a computable function of $k'$.

\begin{theorem}
  \label{the:vertex kernel}
  The minimum sum vertex cover problem has a $(k^2+2 k)$-vertex kernel, and it can be computed in linear time.
\end{theorem}

Similar to the classic Buss reduction for the vertex cover problem~\cite{buss-93-nondeterminism-within-p}, we try to dispose of vertices with degrees higher than~$k$. 
Recall that Buss' reduction removes all high-degree vertices and then isolated vertices, whose neighbors, if any, are all high-degree ones.
For our problem, however, we cannot remove these vertices because we do not know the positions of these high-degree vertices in an optimal solution.
Let~$V_{> k}$ denote the set of vertices of degrees greater than~$k$, and~$I$ the set of isolated vertices in~$G - V_{> k}$.  Since each vertex in~$V_{> k}$ has to be in one of the first~$k$ positions of any solution, $|V_{> k}| \le k$.
We can use Buss' argument to bound non-isolated vertices in~$G - V_{> k}$, and the main work is to deal with~$I$.

We observe that no vertex in~$I$ can be in the first~$k$ positions.
As a result, the edges between~$V_{> k}$ and $I$ are only needed for their ends in~$V_{> k}$.  More specifically, what matters is the number of edges between each vertex in~$V_{> k}$ and $I$.  We can then replace~$I$ with a small set of vertices and move all the edges between~$V_{> k}$ and $I$ to be incident to this set.  The size of this set can be set as the maximum number of neighbors in~$I$ a vertex in~$V_{> k}$ can have, which is upper bounded by the maximum degree of the graph.
We have thus a kernel if we can bound the maximum degree of the graph.
If the degrees of two vertices differ by more than~$k$, i.e., $d(u) > d(v) + k$, then $u$ must appear before~$v$ in any solution.
We sort the vertex degrees non-increasingly.  If the degrees of two consecutive vertices differ by more than~$k$, then we can partition $V(G)$ into~$V_{1}$ and~$V_{2}$ such that $d(u) > d(v) + k$ for every pair of~$u\in V_{1}$ and~$v\in V_{2}$.
Again, the purpose of the edges between~$V_{1}$ and~$V_{2}$ is for their ends in~$V_{1}$.
For each vertex in~$V_{1}$, we remove the same number of edges, chosen arbitrarily, between it and~$V_{2}$.
Together with $|V_{> k}| \le k$, we have a bound $k^2 + k$ on the maximum degree of the graph.

It is worth stressing that we only conduct each reduction once, and thus the algorithm can be carried out in linear time.

Our second result is a parameterized algorithm.  We may start with the kernel.  Note that once we know the set of the first~$k$ vertices in a solution, a straightforward adaption of the Held--Karp algorithm~\cite{held-62-dp} can produce a feasible ordering. 
A trivial approach is thus to enumerate all $k$-vertex subsets of $V(G)$, followed by the Held--Karp algorithm to find an ordering.
The running time is $(2k)^{2k} k^{O(1)}$, which already improves~\cite{AP2024}.  We show an even better result, where $m$ is the size of the input graph.
\begin{theorem}
  \label{thm:alg-msvc}
  There is an $O(m + 2^kk! k^4)$-time algorithm for minimum sum vertex cover.
\end{theorem}

We use a branching algorithm, and hence it suffices to consider yes-instances.
Suppose that~$\sigma$ is a solution, and $S^+$ is the set of first~$k$ vertices in~$\sigma$.
Since $S^+$ is a vertex cover of~$G$, there must be a subset~$S$ of~$S^+$ that is a minimal vertex cover of~$G$.
Damaschke \cite{D2006} showed how to enumerate in $O(m + k^2 2^k)$ time all minimal vertex covers of size at most $k$.
We can thus start with guessing the set~$S$ of vertices and their positions in~$\sigma$.
It remains to fill in the empty positions in the first $k$, by vertices from~$S^+ \setminus S$.
We show that it suffices to try a small number of candidates for each of them.

Before we conclude this section, let us mention that the problem is significantly easier on regular graphs.
For example, it admits an approximation ratio of 1.225~\cite{FLT2004, S2022}, compared to 16/9 on general graphs.
This is also the case for parameterized algorithms.
For a nontrivial regular graph, where the degree is three or more, $\tau \ge n/2$.  Thus, we can find the set of first~$k$ vertices in $O(2^{2k} k^{2})$ time, and then use the Held--Karp algorithm to find the ordering in the same time.
Therefore, the difficulty of the minimum sum vertex cover problem seems to be fundamentally different from the vertex cover problem.

\section{Preliminaries}

All the graphs discussed in this paper are finite and simple. The vertex set and edge set of a graph $G$ are denoted by, respectively, $V(G)$ and~$E(G)$, and $n = |V(G)|$ and~$m=|E(G)|$.  The set of \emph{neighbors} of a vertex~$v$ in~$G$ is denoted as~$N_{G}(v)$, and $d_{G}(v) = |N_{G}(v)|$.  We drop the subscript when the graph is clear from context.
We use $G -S$ to denote the subgraph of~$G$ induced by~$V (G) \setminus S$, which is further shortened to~$G - v$ when $S$ consists of a single vertex~$v$.

A \emph{solution} of a graph~$G$ is a bijection~$\sigma: V(G) \rightarrow \{1, 2, \ldots, n\}$.
Under the solution~$\sigma$, the \emph{cost} of an edge $u v$ is $\min(\sigma(u), \sigma(v))$.  The \emph{maximum cost}  and the \emph{total cost} of the solution are, respectively,
\[
  \max_{uv \in E(G)} \min(\sigma(u), \sigma(v)),
  \quad
  \sum_{uv \in E(G)} \min(\sigma(u), \sigma(v)).
\]
We are now ready to formally define the problem.  Note that we need $w$ in the input as otherwise the problem degenerates to the vertex cover problem.

\begin{itemize}
\item[Input:] a graph~$G$, two nonnegative integers~$k$ and~$w$.
\item[Output:] whether there exists a solution~$\sigma$ of~$G$ such that the maximum cost is at most~$k$ and the total cost is at most~$w$.
\end{itemize}

We say that a solution~$\sigma$ is \emph{feasible} if its maximum cost is at most~$k$ and the total cost is at most~$w$, and \emph{optimal} if its maximum cost is at most~$k$ and the total cost is minimized.  Note that $G$ might admit a solution with a smaller total cost and a maximum cost strictly greater than~$k$.

If $\sigma$ is a feasible solution, then the first $k$ vertices in~$\sigma$ form a vertex cover of~$G$.  

\begin{proposition}\label{lem:high-degree}
  Let~$(G, w, k)$ be a yes-instance and $v$ a vertex of~$G$.
  If $d(v) > k$, then $\sigma(v) \le k$ for any feasible solution~$\sigma$ of $(G, w, k)$. 
\end{proposition}
\begin{proof}
  If $\sigma(v) > k$, then $\sigma(x) \le k$ for all $x\in N(v)$.   This is impossible because $d(v) > k$.
\end{proof}

We use $r_{\sigma}(i)$ to denote the number of edges with cost~$i$, i.e.,
\[
  r_{\sigma}(i) = | \{x\in N(\sigma^{-1}(i))\mid \sigma(x) > i\} |.
\]
Note that the total cost of the solution~$\sigma$ can be alternatively written as $\sum_{i=1}^{n} i\cdot r_{\sigma}(i)$.
In an optimal solution, the numbers $r_{\sigma}(\cdot)$ are non-increasing.

\begin{lemma}
  \label{lem:fix position}
  Let~$(G, w, k)$ be a yes-instance and
  $\sigma$ an optimal solution of $(G, w, k)$.
  For any pair of vertices~$u$ and $v$ of~$G$,
  \begin{enumerate}[i)]
  \item if $\sigma(u) < \sigma(v)$, then $r_{\sigma}(\sigma(u))\ge r_{\sigma}(\sigma(v))$; and
  \item if $d(u) - k \geq d(v) > 0$, then $\sigma(u) < \sigma(v)$.
  \end{enumerate}
\end{lemma}
\begin{proof}
  (i) It suffices to show that if $\sigma(u) = \sigma(v) - 1$, then $r_{\sigma}(\sigma(u))\ge r_{\sigma}(\sigma(v))$.
  Let~$\sigma'$ be obtained from $\sigma$ by switching $u$ and $v$.
  For each edge $u x$ with $\sigma(x) > \sigma(v)$, the cost  is $\sigma(u)$ in~$\sigma$ and~$\sigma(v)$ in~$\sigma'$. 
  For each edge $v x$ with $\sigma(x) > \sigma(v)$, the cost  is $\sigma(v)$ in~$\sigma$ and~$\sigma(u)$ in~$\sigma'$.
  All the other edges have the same cost in~$\sigma$ and~$\sigma'$.
  By the optimality of $\sigma$, we have 
  $r_{\sigma}(\sigma(u)) - r_{\sigma}(\sigma(v)) \le 0$.

  (ii) Noting that $\sigma(u) \le k$ by Proposition~\ref{lem:high-degree}, we have
  \[
    r_{\sigma}(\sigma(u)) \ge d(u) - (\sigma(u) - 1) > d(u) - k \ge d(v) \ge r_{\sigma}(\sigma(v)).
  \]
  Thus, $\sigma(u) < \sigma(v)$.
\end{proof}

\section{Kernelization}

We number the vertices of the input graph as $v_1, v_2, \ldots, v_{n}$ such that $d(v_{i}) \ge d(v_{i+1})$ for all $i = 1, \ldots, n-1$.
We start with dealing with vertices with high degrees.
The safeness of our first rule follows immediately from Proposition~\ref{lem:high-degree}.

\begin{reduction}
  \label{rul:trivial}
  If $d(v_{k +1}) > k$, then return a trivial no-instance.
\end{reduction}

For $i = 1, \ldots, n - 1$, let $\delta_{i} = d(v_i) - d(v_{i+1})$.
The safeness of our second rule is ensured by Lemma~\ref{lem:fix position}.

\begin{reduction}
  \label{rul:max-degree}
  If there exists $t$ such that $\delta_t > k$, 
  then for each $i = 1, \ldots, t$, remove $\delta_t - k$ edges between $v_{i}$ and $\{v_{t+1}, \ldots, v_{n}\}$, and decrease $w$ by ${(t^2+t)}(\delta_t - k)/2$.
\end{reduction}
\begin{proof}[Safeness of Rule~\ref{rul:max-degree}.]
Denote by~$G'$ the resulting graph obtained after applying Rule~\ref{rul:max-degree}, and let $w' = w - {(t^2+t)}(\delta_t - k)/2$.   We show that $(G, w, k)$ is a yes-instance if and only if $(G', w', k)$ is a yes-instance.

Suppose that $(G, w, k)$ is a yes-instance, and let~$\sigma$ be an optimal solution.
Note that $G'$ is a subgraph of~$G$, and the costs of remaining edges remain the same.
By the ordering of the vertices, for every pair of~$i$ and~$j$ with~$i \le t< j$, we have
\[
  d_{G}(v_i) \ge d_{G}(v_t) = d_{G}(v_{t+1}) + \delta_t \geq d_{G}(v_{t+1}) + k \ge d_{G}(v_j) + k.
\]
Since $d_{G}(v_{t+1}) > 0$, we have $\sigma(v_i) < \sigma(v_j)$ by Lemma~\ref{lem:fix position}(ii).
Thus, $\sigma(v_i) \leq t$ if and only if $i \le t$.
The cost of the solution~$\sigma$ of~$G'$ is 
\[
  w - 1 (\delta_t - k) - 2 (\delta_t - k) - \cdots - t (\delta_t - k) =
w - \frac{(t^2+t)(\delta_t - k)}{2} = w'.
\]
Therefore, $\sigma$ is a feasible solution of~$(G', w', k)$.  

Now suppose that $(G', w', k)$ is a yes-instance, and let $\sigma'$ be an optimal solution in which all the isolated vertices of~$G'$ are at the end.
By construction, for all $i = 1, \ldots, t$,
\[
  d_{G'}(v_{i}) = d_{G}(v_{i}) - (\delta_t - k) \ge
  d_{G}(v_{t}) - (\delta_t - k) =
  d_{G}(v_{t+1}) + k \ge
 \max_{j = t+1}^{n} d_{G'}(v_{j}) + k.
\]
Thus, for every pair of~$i$ and~$j$ with~$i \le t< j$, we have $\sigma'(v_i) < \sigma'(v_j)$: by Lemma~\ref{lem:fix position}(ii) when $d_{G'}(v_{j}) > 0$ or by assumption otherwise.
Thus, $\sigma'(v_i) \leq t$ if and only if $i \le t$.
Since all the edges in~$E(G)\setminus E(G')$ are incident to~$\{v_{1}, \ldots, v_{t}\}$, the solution $\sigma'$ is also feasible for~$G$.
The cost of~$\sigma'$ on~$G$ is $w' + {(t^2+t)}(\delta_t - k)/2 = w$.
\end{proof}

Rules~\ref{rul:trivial} and~\ref{rul:max-degree} bound the maximum degree of the graph.
Let~$V_{>k}$ denote the set of vertices of degrees greater than $k$.

\begin{lemma}\label{lem:maximum-degree}
  If neither of Rules~\ref{rul:trivial} and~\ref{rul:max-degree} can be applied, then $d(v_1) \leq k(|V_{>k}|+1)$.
\end{lemma}
\begin{proof}
  It is trivial when $n \le k$, and hence we assume otherwise.
  By assumption,
  $d(v_1) = \sigma_{1} + \cdots + \sigma_{|V_{>k}|} + d(v_{|V_{>k}|+1}) \leq k(|V_{>k}|+1)$.
\end{proof}

We now turn to vertices of degrees at most $k$.
Let~$I$ denote the set of isolated vertices of the subgraph~$G - V_{>k}$.
Note that $N(v)\subseteq V_{>k}$ for all vertices~$v\in I$.  We can use Buss' observation to bound the number of vertices that are not isolated in~$G - V_{>k}$.

\begin{reduction}[\cite{buss-93-nondeterminism-within-p}]
  \label{rul:buss}
  If $|V(G)\setminus (V_{> k}\cup I)| > (k-|V_{>k}|) (k + 1)$, then return a trivial no-instance.
\end{reduction}
\begin{proof}[Safeness of Rule~\ref{rul:buss}]
Let $k' = k - |V_{>k}|$.
If $|V(G)\setminus (V_{> k}\cup I)| > k k' + k'$, then we must assign at least $k' + 1$ vertices from them to a position no later than $k$, which is not possible.
\end{proof}

The next rule is our main tool to get rid of most vertices.
Note that $\max_{v \in V_{> k}}|N(v)\cap I| = |I|$ if and only if $I\subseteq N(v_{i})$ for some $i = 1, \ldots, |V_{> k}|$.
\begin{reduction}
  \label{rul:change vertices in I}
   Let~$p = \max_{v \in V_{> k}}|N(v)\cap I|$.  If $p < |I|$, 
  \begin{enumerate}[i)]
  \item add $p$ new vertices $x_{1}, \ldots, x_{p}$ to $G$;
  \item for each vertex $v \in V_{> k}$, add edges
    $\{v x_{i}\mid 1\le i \le |N(v)\cap I|\}$; and
  \item delete all the vertices in~$I$ from~$G$.
  \end{enumerate}
\end{reduction}

For the safeness of Rule~\ref{rul:change vertices in I}, we need the following lemma.

\begin{lemma}
  \label{lem:I not in opt}
  Let~$(G, w, k)$ be a yes-instance and $v$ a vertex of~$G$.
  If $d(x) > k$ for all the neighbors~$x$ of~$v$, then $\sigma(v) > \sigma(x)$ in any optimal solution~$\sigma$ of $(G, w, k)$.
\end{lemma}
\begin{proof}
  Suppose for contradiction that $X = \{x\in N(v)\mid \sigma(v) < \sigma(x)\}$ is not empty.
  We take the vertex~$u$ from~$X$ such that $\sigma(u)$ is minimized.
  By the selection of~$u$, precisely $r_{\sigma}(\sigma(v)) - 1$ neighbors of~$v$ appear after $u$ in~$\sigma$.
  On the other hand, by Lemma~\ref{lem:fix position}(i), $r_{\sigma}(\sigma(u))\le r_{\sigma}(\sigma(v))$.
  Thus, at least $d(u) - r_{\sigma}(\sigma(v))$ neighbors of~$u$ appear earlier than~$u$ in~$\sigma$.
  Thus, the position of the last neighbor of~$v$ in~$\sigma$ is at least $d(u) - r_{\sigma}(\sigma(v)) + 1 + r_{\sigma}(\sigma(v)) - 1 = d(u) > k$, a contradiction to Proposition~\ref{lem:high-degree}.
\end{proof}
\begin{proof}[Safeness of Rule~\ref{rul:change vertices in I}]
Denote by $G'$ the resulting graph obtained after applying Rule~\ref{rul:change vertices in I}.
Each vertex in~$V_{>k}$ has the same degree in~$G$ and~$G'$, while the degree of a vertex inserted by Rule~\ref{rul:change vertices in I} is at most $|V_{> k}| \le k$.
Thus, in any feasible solution~$\sigma$ of~$(G', w, k)$, we have (1)~$\sigma(v)\le k$ for all $v\in V_{>k}$ by Proposition~\ref{lem:high-degree}, and (2)~$\sigma'(x_{i}) > \sigma'(v)$ for all $i = 1, \ldots, p$ and $v\in N_{G'}(x_{i})$ by Lemma~\ref{lem:I not in opt}.
We can build a bijection between $E(G)\setminus E(G')$ and $E(G')\setminus E(G)$.
Each feasible solution~$\sigma$ of~$(G, w, k)$ has the same cost in~$G'$, and the same holds for each feasible solution of~$(G', w, k)$.  Thus, $(G, w, k)$ and $(G', w, k)$ are equivalent.
\end{proof}

We now summarize our kernelization algorithm in Figure~\ref{alg:kernel}.
Note that it does not change the value of $k$, and after a rule is not applicable, we do not need to check it again.   Thus, it can be carried out in linear time.

\begin{figure}[h!]
  \centering 
  \begin{tikzpicture}
    \path (0,0) node[text width=.7\textwidth, inner xsep=20pt, inner ysep=10pt] (a) {
      \begin{minipage}[t!]{\textwidth}
        \begin{tabbing}
          AAA\=Aaa\=aaa\=Aaa\=MMMMMAAAAAAAAAAAA\=A \kill
          1.\> {\bf if} $k < 0$ {\bf then return} a no-instance;
          \\
          2.  \> $k_0 \leftarrow |\{v \mid d(v)>k\}|$;
          \\
          3.\> {\bf if} $k_0 > k$ {\bf then return} a no-instance (Rule~\ref{rul:trivial});
          \\
          4.\> number the vertices such that $d(v_{1}) \ge d(v_{2}) \ge \cdots \ge d(v_{n})$;
          \\
          5.\> {\bf if} $d(v_1) > k(k_0 + 1)$ {\bf then} apply Rule~\ref{rul:max-degree};
          \\
          6.\> apply Rule~\ref{rul:buss} and Rule~\ref{rul:change vertices in I};
          \\
          7.\> {\bf return} $(G, w, k)$.
        \end{tabbing}
      \end{minipage}
    };
    \draw[draw=gray!60] (a.north west) -- (a.north east) (a.south west) -- (a.south east);
  \end{tikzpicture}
  \caption{The kernelization algorithm for minimum sum vertex cover.}
  \label{alg:kernel}
\end{figure}

\begin{proof}[Proof of Theorem~\ref{the:vertex kernel}]
We use the algorithm described in Figure~\ref{alg:kernel}.
The correctness of step~1 is trivial.
We have seen the safeness of the reduction rules, applied in steps~2--5.
Note that Rule~\ref{rul:max-degree} must be applicable in step~5 by Lemma~\ref{lem:maximum-degree}.
After step~6, the number of vertices is at most
\[
  k_{0} + (k-k_0) (k + 1) + k (k_0+1) = k^2 + 2 k.
\]
Since we process each edge once, the algorithm can be carried out in $O(n + m)$ time.
\end{proof}

\section{A parameterized algorithm}

We present a simple branching algorithm for the minimum sum vertex cover problem.
Suppose that $(G, w, k)$ is a yes-instance,
with a solution~$\sigma$.
By definition, the set of first~$k$ vertices in~$\sigma$ is a vertex cover of~$G$, and there must be a subset~$S$ of it is a minimal vertex cover of~$G$.
We use branching to find vertices in this minimal vertex cover and their positions in~$\sigma$, followed by other vertices.
Damaschke \cite{D2006} showed how to enumerate in $O(m + k^22^k)$ time all minimal vertex covers of size at most $k$.

\begin{lemma}[\cite{D2006}]
  A graph has at most $2^k$ minimal vertex covers of size at most $k$, and we can enumerate them in $O(m + k^2 2^k)$ time.
\end{lemma}

For each minimal vertex cover of~$G$, we use branching to find positions of these vertices.
In one of the branchings, the minimal vertex cover is~$S$, and the position of a vertex~$v\in S$ is $\sigma(v)$.
It remains to fill in the gaps with vertices from~$V(G) \setminus S$.
Let~$i$ be such a position.
The purpose of including a vertex~$v\in V(G)\setminus S$ in position~$i$ is that it decreases the costs of edges between~$v$ and vertices in later positions.
Note that there are no edges among~$V(G)\setminus S$, the ends of all these edges are already in position.
For each vertex in~$V(G)\setminus S$, we can calculate how much the total cost decreases with perching it at position~$i$.
We choose the $k - |S|$ ones with the top scores, and we show that it is always safe to use one of them.
We need $k - |S|$ here because these vertices may be needed by other gaps.

We are now ready to summarize the algorithm in Figure~\ref{alg:main} and use it to prove Theorem~\ref{thm:alg-msvc}.

\begin{figure}[h!]
  \centering 
  \begin{tikzpicture}
    \path (0,0) node[text width=.75\textwidth, inner xsep=20pt, inner ysep=10pt] (a) {
      \begin{minipage}[t!]{\textwidth}
        \begin{tabbing}
          AAA\=Aaa\=aaa\=Aaa\=MMMMMAAAAAAAAAAAA\=A \kill
          1. \> guess a minimal vertex cover $S$ of~$G$ with $|S| \le k$;
          \\
          2. \> guess an injective mapping $\sigma: S \rightarrow \{1,2, \ldots,k\}$;
          \\
          3. \> {\bf for each} unoccupied position~$p < k$ of~$\sigma$ \textbf{do}
          \\
          3.1. \>\> {\bf for each} $u \in V(G)\setminus S$ {\bf do}
          \\
          3.1.1. \>\>\> $s(u) \leftarrow \sum \{j - p\mid \sigma^{-1}(j)\in N(u), j > p\}$;
          \\
          3.2. \>\> $C\leftarrow k - |S|$ vertices maximizing $s(\cdot)$, breaking ties arbitrarily;
          \\
          3.3. \>\> guess a vertex $x\in C$, add $x$ to~$S$, and $\sigma(x) \leftarrow p$;
        \\
        4.\> {\bf if} the cost of $\sigma$ is $\le w$ {\bf then return} ``yes'';
        \\
        \> {\bf else return} ``no.''
        \end{tabbing}
      \end{minipage}
    };
    \draw[draw=gray!60] (a.north west) -- (a.north east) (a.south west) -- (a.south east);
  \end{tikzpicture}
  \caption{The parameterized algorithm for minimum sum vertex cover.}
  \label{alg:main}
\end{figure}

\begin{proof}[Proof of Theorem~\ref{thm:alg-msvc}]
  We use the algorithm described in Figure~\ref{alg:main}.
  Since the algorithm returns ``yes'' only when a solution is identified, it suffices to show a solution must be returned for a yes-instance.
  We fix an optimal solution~$\sigma^*$, and we may number the vertices such that $\sigma^*(v_{i}) = i$.
  Let~$S^* = \{v_{1}, \ldots, v_{k}\}$.  Since $\sigma^*$ is feasible, $S^*$ is a vertex cover of~$G$.
  Step~1 finds a minimal vertex cover~$S$ of~$G$ that is a subset of~$S^*$, and step~2 finds $\sigma(v)$ of each vertex~$v\in S$.
  Note that $\sigma^* = \sigma$ when $|S| = k$, and we are already done.
  Otherwise, for each unoccupied position $p$, step~3 select a vertex $u$ from $V(G)\setminus S$ and assign $u$ to~$p$.

  Step~3.3 finds the vertex~$v_{p}$ if it is in~$C$.  Now suppose that $v_{p}\not \in C$.
  By the selection of vertices in~$C$, it contains a vertex $v_{i}$ with $i >k$ and $s(v_{i}) \geq s(v_{p})$.
  Let~$\sigma'$ be the solution obtained from~$\sigma^*$ by switching~$v_{i}$ and~$v_{p}$.
  We argue that $\sigma'$ has the same cost as~$\sigma^*$, hence also an optimal solution.
  Between $\sigma'$ and $\sigma^*$, only the costs of edges incident to~$v_{i}$ and~$v_{p}$ can be different.
  The cost of the edge~$v_{i} v_{j}$ or $v_{p} v_{j}$ remains $j$ if $j < p$.  There is no edge between $v_{i}$ and~$v_{j}$ or between $v_{p}$ and~$v_{j}$ for $j > k$ (because $S$ is a vertex cover).
  Thus, the difference between the total costs of $\sigma'$ and $\sigma^*$ is $s(v_{p}) - s(v_{j}) \le 0$.
Thus, the algorithm always ends with a solution with the same total cost as~$\sigma^*$ (note that $s(v_{p}) - s(v_{j}) \ge 0$ since $\sigma^*$ is an optimal solution).

We now analyze the running time.
We may preprocess the instance with the kernelization algorithm, and hence $G$ has at most $O(k^2)$ vertices.
We can use the algorithm of Damaschke \cite{D2006} to find all the required minimal vertex covers in $O(m + k^22^k)$ time.
For each of them, the number of injective mappings is $k!/(k- |S|)!$.
Steps~3.1 and~3.2 can be done in $O(k^4)$ time.  Step~3.3 makes $k - |S|$ branches.  Note that $|S|$ increases by one with each iteration, and thus the total number of branching is $(k - |S|)!$.
Step~4 takes $O(k^4)$ time.  The total time is thus $O(m + 2^kk! k^4)$.
\end{proof} 

\section{Appendix} 

\begin{proof}[Proof of Lemma~\ref{lem:mvc and msvc}]
  Let~$S^*$ be a minimum vertex cover of~$G$, and~$\sigma$ an optimal solution of $(G, w, k)$.
  We number the vertices such that $\sigma(v_{i}) = i$.  We may assume that $v_{1}\not\in S^*$: otherwise we consider the subgraph $G - v_{1}$ (note that $S^*\setminus \{v_{1}\}$ is a minimum vertex cover of $G - v_{1}$, and $\langle v_{2}, v_{3}, \ldots, v_{n}\rangle$ is an optimal solution of $(G - v_{1}, w - m, k - 1$)).
  As a result,
  \[
    r_{\sigma}(i) \le r_{\sigma}(1) = d(v_{1}) \le \tau
  \]
  by Lemma~\ref{lem:fix position}(i) and that $N(v_{1})\subseteq S^*$.

  Let~$k$ be the maximum cost of an edge in~$\sigma$, and  $c = \lfloor (m - k)/(\tau -1) \rfloor$.
  The minimum possible cost of~$\sigma$ is achieved when 
  \[
    r_{\sigma}(i) =
    \begin{cases}
      \tau & i \le c,
      \\
      m - k + 1 - c(\tau - 1) & i = c+1,
      \\
      1 & i > c+1.
    \end{cases}
  \]
  Thus, the cost of $\sigma$ is at least
  \begin{align*}
    &\tau (1 + \cdots + c) + ((c+1) + \cdots + k)
    \\
    =& {c(c+1) (\tau - 1)\over 2} + {k (k+1)\over 2}
    \\
    \ge& {m + k^2 \over 2} + {(m - k)^2 \over 2(\tau - 1)}.
  \end{align*}

  We take an ordering~$\pi$ of~$V(G)$ such that $\pi(x) \le \tau$ if and only if $x\in S^*$, and
\[
  r_{\pi}(1) \ge r_{\pi}(2) \ge \cdots \ge r_{\pi}(\tau).
\]
The maximum possible cost of $\pi$ is when $r_\pi(i)$ is either~$\lceil {m \over \tau}\rceil$ or~$\lfloor {m \over \tau}\rfloor$.  Thus, the cost of $\pi$ is at most $m (\tau+1)/2$.  The optimality of~$\sigma$ implies
\[
  {m + k^2 \over 2} + {(m - k)^2 \over 2(\tau - 1)} \le {m (\tau+1)\over 2}.
\]
Thus,
\[
  k\le \sqrt{m (\tau - 1) - {m^2\over \tau} + {m^2\over \tau^2}} + {m\over \tau} = O(\tau^{1.5}).
\]
This concludes the proof.
\end{proof}

It is worth stressing that the bound in the proof of Lemma~\ref{lem:mvc and msvc} cannot be attained.  In particular, if $r_\sigma(k') > 0$, then $r_\sigma(k' - \tau) > 1$ (i.e., there cannot be more than $\tau + 1$ vertices that each effectively covers one edge): we can always replace $\{v_{k' - \tau}, \ldots, v_{k'}\}$ with~$S^*$ to produce a strictly smaller solution.

\begin{conjecture}
  Any graph has an optimal solution in which the maximum cost of an edge is $\tau+ O(\log \tau)$.
\end{conjecture}

The Held--Karp algorithm can be adapted to solve the problem in $O(2^n n^2)$ time.

\begin{theorem}
  \label{the: G is regular}
  There is an $O(4^k k^2)$-time algorithm for the minimum sum vertex cover problem on regular graphs.
\end{theorem}
\begin{proof}
  The problem is trivial when the degree is $0$ or $1$.
  Otherwise, $k \ge \tau \ge n/2$.  Thus, the problem can be solved in $O(4^k k^2)$ time.
\end{proof}

\bibliographystyle{plainurl}
\bibliography{main}
\end{document}